\documentclass[journal,twocolumn]{IEEEtran}
\usepackage{amsfonts}
\usepackage{times}
\usepackage{graphicx}
\usepackage{latexsym}
\usepackage{dsfont}
\usepackage{amssymb}
\usepackage{amsmath}
\usepackage{cite}
\usepackage{verbatim}
\usepackage{subfigure}

\newtheorem{assumption}{Assumption}

\newtheorem{proposition}{Proposition}

\newenvironment{proof}{ \textbf{Proof:} }{ \hfill $\Box$}

\newcommand{\figref}[1]{{Fig.}~\ref{#1}}

\def\bb0{{\mathbb{0}}}

\def\ba{{\mathbf{a}}}
\def\bb{{\mathbf{b}}}

\def\bff{{\mathbf{f}}}
\def\bg{{\mathbf{g}}}
\def\bh{{\mathbf{h}}}

\def\bm{{\mathbf{m}}}

\def\bv{{\mathbf{v}}}

\def\bx{{\mathbf{x}}}

\def\b0{{\mathbf{0}}}

\def\bA{{\mathbf{A}}}
\def\bB{{\mathbf{B}}}

\def\bI{{\mathbf{I}}}

\def\bR{{\mathbf{R}}}

\def\bW{{\mathbf{W}}}

\def\bbE{{\mathbb{E}}}

\def\cA{\mathcal{A}}

\def\cN{\mathcal{N}}

\def\sf0{{\mathsf{0}}}

\usepackage{epstopdf}
\usepackage{enumerate}
\usepackage{algorithmicx}
\usepackage{algorithm}
\usepackage{amsmath}
\usepackage[noend]{algpseudocode}
\usepackage{float}
\usepackage{hyperref}
\usepackage{color}
\usepackage{makeidx}
\usepackage{bbm}
\usepackage{graphicx}
\usepackage{multirow}
\usepackage{balance}

\newcommand{\sref}[1]{{Section}~\ref{#1}}

\newcommand{\pinv}[1]{\ensuremath{#1^{\dagger}}} 	

\def\bsx{\boldsymbol{x}}

\begin{document}
\title{Deep Learning for TDD and FDD Massive MIMO: Mapping Channels in Space and Frequency}
\author{Muhammad Alrabeiah and Ahmed Alkhateeb\\  {Arizona State University, Email: $\{$malrabei, aalkhateeb$\}$@asu.edu}}
\maketitle

\begin{abstract}

Can we map the channels at one set of antennas and one frequency band to the channels at another set of antennas---possibly at a different location and a different frequency band? If this channel-to-channel mapping is possible, we can expect dramatic gains for massive MIMO systems. For example, in FDD massive MIMO, the uplink channels can be mapped to the downlink channels or the downlink channels at one subset of  antennas can be mapped to the downlink channels at all the other antennas. This can significantly reduce  (or even eliminate) the downlink training/feedback overhead. In the context of cell-free/distributed massive MIMO systems, this channel mapping can be leveraged to reduce the fronthaul signaling overhead as  only the channels at a subset of the distributed terminals need to be fed to the central unit which can map them to the channels at all the other terminals. This mapping can also find interesting applications in mmWave beam prediction, MIMO radar, and massive MIMO based positioning. 

In this paper, we introduce the new concept of \textit{channel mapping in space and frequency}, where the channels at one set of antennas and one frequency band are mapped to the channels at another set of antennas and  frequency band. First, we prove that this channel-to-channel mapping function exists under the condition that the mapping  from the candidate user positions to the channels at the first set of antennas is bijective; a condition that can be achieved with high probability in several practical MIMO communication scenarios. Then, we note that the channel-to-channel mapping function, even if it exists, is typically unknown and very hard to characterize analytically as it heavily depends on the various elements of the surrounding environment. With this motivation, we propose to leverage the  powerful learning capabilities of deep neural networks to learn (approximate) this complex channel mapping function. For a case study of distributed/cell-free massive MIMO system with 64 antennas, the results show that acquiring the channels at only 4-8  antennas can be efficiently mapped to the channels at all the 64 distributed antennas, even if the 64 antennas are at a different frequency band.  Further, the 3D ray-tracing based simulations show that the achievable rates with the predicted channels achieve near-optimal data rates when compared to the upper bound with perfect channel knowledge. This highlight a novel solution for reducing the training and feedback overhead in mmWave and massive MIMO systems thanks to the powerful learning capabilities of deep neural networks.

 \end{abstract}

\section{Introduction} \label{sec:Intro}
 
 Scaling the number of antennas up is a key characteristic of current and future wireless systems \cite{HeathJr2016,Boccardi2014,11ad,Andrews2014,Sanguinetti2019}. Harvesting the multiplexing and beamforming gains of the large numbers of antennas, however, requires the channel knowledge at these antennas. This is normally associated with large training and feedback overhead that can limit the scalability of massive MIMO systems \cite{Bjoernson2016,HeathJr2016}. In this paper, we ask an important question: {If we know the channels between a user and a certain set of antennas at one frequency band, can we map this knowledge to the channels at a different set of antennas and at a different frequency band?} If such mapping exists and we can characterize or model it, this can yield significant gains for both co-located and distributed massive MIMO systems. Essentially, this mapping means that we can directly predict the downlink channels from the uplink channels, eliminating the downlink training and feedback overhead in co-located/distributed FDD massive MIMO systems. In the context of TDD cell-free massive MIMO, this mapping implies that the channels at only a subset of the distributed terminals need to be fed to the central processing unit that will use them to predict the channels at all the other terminals, which reduces the fronthaul signaling overhead and allows these distributed systems to scale. Motivated by these interesting gains, this paper investigates the existence and modeling of this channel-to-channel mapping function at both space and frequency.

 \begin{figure}[t]
 	\centering
 	\includegraphics[width=.95\linewidth]{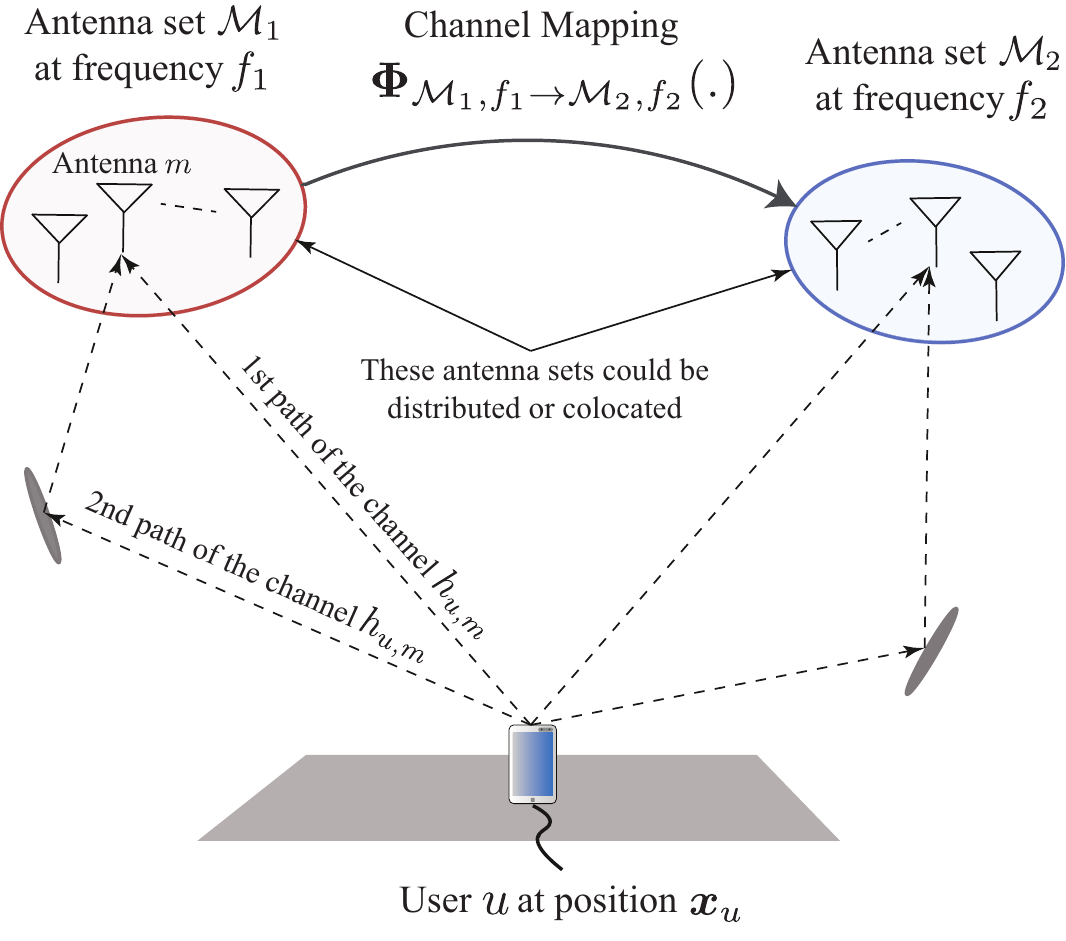}
 	\caption{The figure illustrates the general system model we adopt in this paper. The user is communicating with one of the two sets of antennas, $\mathcal M_1$ at frequency $f_1$ or $\mathcal M_2$ or frequency $f_2$. The function $\boldsymbol{\Phi}_{\mathcal M_1, f_1\rightarrow\mathcal M_2,f_2}(.)$ maps the channels at set $\mathcal M_1$ to those at set $\mathcal M_2$. }
 	\label{fig:map}
 \end{figure}

 Estimating the channels at one frequency band using the channel knowledge at a different frequency band is attracting increasing interest \cite{Vasisht2016,Rottenberg2019,Ali2018,Maschietti2019}. In \cite{Vasisht2016}, the parameters of the uplink channels, such as the angles of arrival and path delays were estimated and used to construct the downlink channels at an adjacent frequency band.  This frequency extrapolation concept was further studied in \cite{Rottenberg2019} where lower bounds on the mean squared error of the extrapolated channels were derived. On a relevant line of work, \cite{Ali2018,Maschietti2019} proposed to leverage the channel knowledge at one frequency band (sub-6 GHz) to reduce the training overhead associated with design the mmWave beamforming vectors, leveraging the spatial correlation between the two frequency bands. A common feature of all the prior work in \cite{Vasisht2016,Rottenberg2019,Ali2018,Maschietti2019} is the requirement to first estimate some spatial knowledge (such as the angles of arrival) about the channels in one frequency band and then leverage this knowledge in the other frequency band. This channel parameter estimation process, however, is fundamentally limited by the system and hardware capability in resolving these channel parameters, which highly affects the quality of the extrapolated channels.

In this paper, we introduce the novel concept of \textit{channel mapping in space and frequency} that does not just translate the channel knowledge from one to another frequency band but also from one set of antennas to another set of co-located or distributed antennas. First, we prove that such channel-to-channel mapping function exists  under the condition that the mapping function from the candidate user positions to the channels at the first set of antennas is bijective. Then, we show that deep neural networks can be leveraged to learn (approximate) this complex channel mapping function if a sufficient number of observations is available. These general results have the potential of significantly reducing the channel training/feedback overhead at both co-located and distributed/cell-free FDD massive MIMO systems. In fact, leveraging such channel-to-channel mapping functions can find interesting applications in other systems as well, such as reducing the fronthaul signaling overhead in TDD cell-free massive MIMO systems and reducing the beam training overhead in mmWave MIMO applications \cite{Alkhateeb2018,Li2018a,Alkhateeb18a,Taha2019}. Simulation results, based on accurate ray-tracing, confirm the potential gains of the proposed deep learning based channel mapping approach. For example, for a given setup with 64 distributed antennas, the results show that acquiring the channels at only 8 of these antennas at an uplink frequency of $2.4$GHz can be efficiently mapped to the $2.5$GHz downlink channels at all the 64 antennas. Adopting conjugate beamforming with the predicted downlink channels, the results show that near-optimal data rates can be achieved when compared with the upper bound with perfect channel knowledge. This highlights a novel solution for reducing the training and feedback overhead in mmWave and massive MIMO systems thanks to the powerful learning capabilities of deep neural networks.

\textbf{Notation}: We use the following notation throughout this paper: $\bA$ is a matrix, $\ba$ is a vector, $a$ is a scalar, and $\cA$ is a set. $|\bA|$ is the determinant of $\bA$, $\|\bA \|_F$ is its Frobenius norm, whereas $\bA^T$, $\bA^H$, $\bA^*$, $\bA^{-1}$, $\pinv{\bA}$ are its transpose, Hermitian (conjugate transpose), conjugate, inverse, and pseudo-inverse respectively. $[\bA]_{r,:}$ and $[\bA]_{:,c}$ are the $r$th row and $c$th column of the matrix $\bA$, respectively. $\mathrm{diag}(\ba)$ is a diagonal matrix with the entries of $\ba$ on its diagonal. $\bI$ is the identity matrix and $\mathbf{1}_{N}$ is the $N$-dimensional all-ones vector. $\bA \otimes \bB$ is the Kronecker product of $\bA$ and $\bB$, and $\bA \circ \bB$ is their Khatri-Rao product. $\cN(\bm,\bR)$ is a complex Gaussian random vector with mean $\bm$ and covariance $\bR$. $\bbE\left[\cdot\right]$ is used to denote expectation.

\section{System and Channel Models} \label{sec:Sys_Model}
Consider the general system model in \figref{fig:map} where one user at position $\bsx_u$ can communicate with one of two candidate sets of antennas, namely $\mathcal{M}_1$ over a frequency $f_1$ and $\mathcal{M}_2$ over a frequency $f_2$.  
It is important to emphasize here that we do not impose any constraints on the relation between the two antenna sets, $\mathcal{M}_1$ and $ \mathcal{M}_2$, nor the frequencies $f_1$ and $f_2$. Therefore, this general system model includes several special cases such as (i) the case when some antennas are common between the two antenna sets, or (ii) when the two antenna sets use the same frequency $f_1=f_2$. This allows us to use the results derived based on the adopted general model to draw important insights  for both co-located and distributed (cell-free) massive MIMO systems and for both TDD and FDD system operation modes, as we will discuss in \sref{sec:Exist}.

\textbf{Channel Model:} Let $h_{u,m}(f_1)$ denote the channel from user $u$ to antenna $m$ in the antenna set $\mathcal{M}_1$ at the frequency $f_1$. Assume that this propagation channel consists of $L$ paths. Each path $\ell$ has a distance $d_\ell$ from the user to antenna $m$, a delay $\tau_\ell$, and a complex gain $\alpha_\ell=|\alpha_\ell| e^{j \phi_\ell}$. The channel $h_{u,m}(f_1)$ can then be written as 
\begin{equation} \label{eq:channel}
h_{u,m}(f_1)=\sum_{\ell=1}^L |\alpha_\ell| e^{j \phi_\ell} e^{- j 2 \pi f_1 \tau_\ell}.
\end{equation}

Note that the magnitude of the path gain $|\alpha_\ell| $ of path $\ell$ depends on (i) the distance $d_\ell$ that this path travels from the user to the scatterer(s) ending at the receiver, (ii) the frequency $f_1$, (iii) the transmitter and receiver antenna gains, and (iv) the cross-section and dielectric properties of the scatterer(s).  The phase $\phi_\ell$ also depends on the scatterer(s) materials and wave incident/impinging angles at the scatterer(s).  Finally, the delay $\tau_\ell=\frac{d_\ell}{c}$, where $c$ is the speed of light. By reciprocity, we  consider $h_{u,m}(f_1)$ as also the downlink channel from antenna $m$ to user $u$.   Now, we define the $\left|\mathcal{M}_1\right| \times 1$ channel vector $\bh_{u, \mathcal{M}_1}(f_1)=[h_{u,1}(f_1), ..., h_{u,\left|\mathcal{M}_1\right|}(f_1)]^T$ as the channel vector from user $u$ to the antennas in set $\mathcal{M}_1$. Similarly, we define the channel vector $\bh_{u, \mathcal{M}_2}(f_2)$ for the channel between user $u$ and the antennas in set $\mathcal{M}_2$.  In the next section, we define the channel mapping problem.

\section{Channel Mapping Problem Definition} \label{sec:Problem}

In this paper, we are interested in answering the following question: {If the channels between the user $u$ and the antenna set $\mathcal{M}_1$ are given, can we use them to estimate the channels between the user and the other set of antennas $\mathcal{M}_2$?} In other words, {can we map $\bh_{u, \mathcal{M}_1}(f_1)$ to $\bh_{u, \mathcal{M}_2}(f_2)$}? 

If such mapping exists and we can characterize/model it, this can have great gains for massive MIMO systems. Essentially, this means that we need only to estimate the channels at one set of antennas and use them to directly predict the channels at all the other antennas for the same frequency or even a different frequency. This can dramatically reduce the training and feedback overhead in FDD massive MIMO systems. Further, for cell-free massive MIMO systems where a large number of distributed antennas are connected to the central unit, leveraging this mapping means that we need only to feed-forward a subset of the antenna channels to the central unit that maps them to the channels at  the rest of the antennas. This may significantly reduce the control overhead on the front-haul of the cell-free massive MIMO systems. With this motivation, we investigate the existence and modeling of this channel mapping function in this paper.  More formally, defining the channel mapping $\boldsymbol{\Phi}_{\mathcal{M}_1, f_1 \rightarrow \mathcal{M}_2, f_2}(.)$ as 
\begin{equation}
\boldsymbol{\Phi}_{\mathcal{M}_1, f_1 \rightarrow \mathcal{M}_2, f_2}: \left\{\bh_{u, \mathcal{M}_1}(f_1)\right\} \rightarrow \left\{\bh_{u, \mathcal{M}_2}(f_2) \right\},
\end{equation}
then we investigate the following two problems:
\begin{equation}
\hspace{-10pt}\textbf{Problem 1:} \ \textit{Does the mapping  $\boldsymbol{\Phi}_{\mathcal{M}_1, f_1 \rightarrow \mathcal{M}_2, f_2}(.)$ exist?} \nonumber
\end{equation}
\begin{equation}
\hspace{-4pt}\textbf{Problem 2:} \ 	\textit{If $\boldsymbol{\Phi}_{\mathcal{M}_1, f_1 \rightarrow \mathcal{M}_2, f_2}(.)$ exists, how to model it?} \nonumber
\end{equation}
 In the following section, we investigate the existence of the channel mapping function then we show in \sref{sec:DL_Mapping} that deep learning can be efficiently leveraged to model it.

\section{The Existence of Channel Mapping} \label{sec:Exist}
In this section, we consider the system and channel models in \sref{sec:Sys_Model} and examine the existence of the channel mapping function $\boldsymbol{\Phi}_{\mathcal{M}_1, f_1 \rightarrow \mathcal{M}_2, f_2}(.)$. First, we investigate the existence of the position to channel and channel to position mapping functions which will lead to the channel to channel mapping.

 \textbf{Existence of position to channel mapping:}
Consider the channel model in \eqref{eq:channel}, where the channel from user $u$ at position $\bsx_u$ to antenna $m$ is completely defined by the parameters $|\alpha_\ell|, \phi_\ell, \tau_\ell$ of each path and the frequency $f_1$. We note that these parameters, $|\alpha_\ell|, \phi_\ell, \tau_\ell$, are functions of the environment geometry, scatterer materials, the frequency $f_1$, in addition to the antenna and user positions, as explained in the discussion after \eqref{eq:channel}. Therefore, for a given static communication environment (including the geometry, materials, antenna positions, etc.), there exists a deterministic mapping function from the position $\bsx_u$ to the channel $h_{u,m}(f_1)$ at every antenna element $m$ \cite{Vieira2017}.  More formally, if $\left\{\bsx_u\right\}$ represents the set of all candidate user positions, with the sets $\left\{\bh_{u,\mathcal{M}_1}(f_1)\right\}$ and $\left\{\bh_{u,\mathcal{M}_2}(f_2)\right\}$ assembling the corresponding channels at antenna sets $\mathcal{M}_1$ and $\mathcal{M}_2$, then we define the position-to-channel mapping functions $\bg_{\mathcal{M}_1,f_1}(.)$ and $\bg_{\mathcal{M}_2,f_2}(.)$ as 
\begin{align}
\bg_{\mathcal{M}_1,f_1}&: \left\{\bsx_u\right\} \rightarrow \left\{\bh_{u,\mathcal{M}_1}(f_1)\right\}, \\
\bg_{\mathcal{M}_2,f_2}&: \left\{\bsx_u\right\} \rightarrow \left\{\bh_{u,\mathcal{M}_2}(f_2)\right\}.
\end{align}

Note that these deterministic mapping functions for a given communication environment scene can be numerically computed (or approximated) using ray-tracing simulations. However, we emphasize here that while we assume the existence of these position-to-channel mapping functions, our developed deep learning solutions will not require the knowledge of these mapping functions, as will be explained in \sref{sec:DL_Mapping}.

\textbf{Existence of channel to position mapping:} 
Next, we investigate the existence of the mapping from the channel vector $\bh_{u,\mathcal{M}_1}(f_1)$ of the antenna set $\mathcal{M}_1$ to the user position. For that, we adopt the following assumption. 
\begin{assumption}
The position-to-channel mapping function, $\bg_{\mathcal{M}_1,f_1}: \left\{\bsx_u\right\} \rightarrow \left\{\bh_{u,\mathcal{M}_1}(f_1)\right\}$, is bijective.    
\label{assumption1}	
\end{assumption}

This assumption means that every user position in the candidate set $\left\{\bx_u\right\}$ has a unique channel vector $\bh_{u,\mathcal{M}_1}(f_1)$. Here, we note that the bijectiveness of this mapping, $\bg_{\mathcal{M}_1,f_1}$, depends on several factors including (i) the number and positions of the antennas in the set $\mathcal{M}_1$, (ii) the set of candidate user locations, and (iii) the geometry and materials of the surrounding environment. While it is hard to guarantee the bijectiveness of $\bg_{\mathcal{M}_1,f_1}(.)$, it is important to note that this mapping is actually bijective with high probability in many practical wireless communication scenarios \cite{Vieira2017}.

Now, we define the channel-to-position mapping function $\bg^{-1}_{\mathcal{M}_1,f_1}(.)$ as the inverse of the mapping $\bg_{\mathcal{M}_1,f_1}(.)$, i.e., 
\begin{equation}
\bg^{-1}_{\mathcal{M}_1,f_1}:   \left\{\bh_{u,\mathcal{M}_1}(f_1)\right\} \rightarrow  \left\{\bsx_u\right\}
\end{equation}
Under Assumption 1, the inverse mapping, $\bg^{-1}_{\mathcal{M}_1,f_1}(.)$, exists. In fact, this inverse mapping is widely adopted in the wireless positioning and fingerprinting literature \cite{Vieira2017,Savic2015}. 

\textbf{Existence of channel to channel mapping:}
After investigating the existence of the position-to-channel and channel-to-position mapping functions, we are now ready to make the following proposition on the existence of the channel-to-channel mapping function, $\boldsymbol{\Phi}_{\mathcal{M}_1, f_1 \rightarrow \mathcal{M}_2, f_2}(.)$. 
\begin{proposition} \label{prop1}
	For a given communication environment, and if assumption 1 is satisfied, then there exists a channel-to-channel mapping function, $\boldsymbol{\Phi}_{\mathcal{M}_1, f_1 \rightarrow \mathcal{M}_2, f_2}(.)$, that equals
	\begin{align}
		\boldsymbol{\Phi}_{\mathcal{M}_1, f_1 \rightarrow \mathcal{M}_2, f_2} & = \bg_{\mathcal{M}_2, f_2} \circ \bg^{-1}_{\mathcal{M}_1, f_1}  \nonumber \\
		& : \left\{\bh_{u,\mathcal{M}_1}(f_1)\right\} \rightarrow \left\{\bh_{u,\mathcal{M}_2}(f_2)\right\}
	\end{align}
\end{proposition}
\begin{proof}
	The proof follows from (i) the existence of the  mapping $\bg^{-1}_{\mathcal{M}_1, f_1}(.)$  under assumption 1, (ii) the existence of the mapping $\bg_{\mathcal{M}_2, f_2}(.)$ for any given environment, and (iii) the existence of the composite function $   \bg_{\mathcal{M}_2, f_2} \circ \bg^{-1}_{\mathcal{M}_1, f_1}(.)$ since the domain of $ \bg_{\mathcal{M}_2, f_2}(.)$ is the same as the co-domain of $\bg^{-1}_{\mathcal{M}_1, f_1}(.)$, and they both equal to $\left\{\bsx_u\right\}$.  
\end{proof}

Consider a communication setup with a base station employing multiple antennas (co-located or distributed), Proposition \ref{prop1} means that once we identify a subset of these antennas that satisfy the bijectiveness condition in Assumption \ref{assumption1}, then there exists a way (mapping function) that can map the channels at this set of antennas to the channels at all other antennas, even if they are communicating at a different frequency. This result has interesting gains for both co-located and distributed massive MIMO systems. Next, we discuss some of these gains in more detail.
\begin{itemize}
	\item \textbf{FDD Co-located and Distributed Massive MIMO:} Note that the general setup adopted in this section and illustrated in \figref{fig:map} reduces to the special case of FDD massive MIMO systems when $\mathcal{M}_1 \subseteq \mathcal{M}_2 $ and when $f_1$ and $f_2$ represent the uplink and downlink frequencies. In this case, Proposition \ref{prop1} implies that only a subset $\mathcal{M}_1$ of the base station antennas need to be trained in the uplink. The uplink channels at these antennas can be directly mapped to the downlink channels at \textit{all} the antennas, which significantly reduces the training and feedback overhead in these systems. We will illustrate these gains in \sref{sec:Results}. 
	
	It is worth noting here that this result maps the channels at both space and frequency. Therefore, it includes the following two  special cases when only mapping in space or frequency is applied.
	\begin{enumerate}[I]
		 \item \textit{Mapping channels in space} is when $\mathcal{M}_1 \subseteq \mathcal{M}_2$ and  $f_1=f_2$ are both representing the downlink frequency . In this case, Proposition  \ref{prop1}  means that only a few antennas need to be trained in the downlink and the rest can be constructed by channel mapping. For example, consider a basestation with 100 antennas. If $5$ antennas are enough to satisfy the bijectiveness condition in Assumption 1, then we only need to downlink train and feedback $5$ antennas instead of the $100$ antennas, which is a missive reduction in the training/feedback overhead. 
		\item  \textit{Mapping channels in frequency} is when  $\mathcal{M}_1=\mathcal{M}_2$ and  $f_1$, $f_2$  represent the uplink and downlink channels. In this case, Proposition  \ref{prop1} means that the uplink channels can be directly mapped to the downlink channels which completely eliminates the downlink training/feedback overhead. 
	\end{enumerate} 
	
	\item \textbf{TDD Distributed (Cell-free) Massive MIMO:} In TDD cell-free massive MIMO systems, the distributed antenna terminals estimate the uplink channels and use it for the downlink transmission. To avoid the need for forwarding all the uplink channels from the terminals to the central processing unit, the initial proposals for cell-free massive MIMO systems adopted conjugate beamforming where every terminal independently designs  its downlink precoding weight. If feeding forward all the channels to the central processing is feasible, then several gains can be achieved, such as the ability to adopt more sophisticated precoding strategies and advanced user scheduling among others. Feeding forward all the channels to the central processing unit, however, is associated with high overhead that can limit the scalability of cell-free massive MIMO systems. Interestingly, Proposition \ref{prop1} suggests that only a subset $\mathcal{M}_1 \subseteq \mathcal{M}_2$ of these antennas  need to be forward their channels to the central unit which can map them to the channels at all the other antennas. This has the potential of significantly reducing the channel feed-forward overhead  with and  enabling the scalability of cell-free massive MIMO systems. We will further highlight these gains in \sref{sec:Results}.
\end{itemize}

It is also worth mentioning that the channel mapping result in Proposition \ref{prop1} can  also have several interesting applications in mmWave systems, such as using the channels collected at a few distributed antennas to predict the best beam at an antenna array \cite{Alkhateeb2018}, or using sub-6GHz channels to predict the mmWave blockages and beamforming vectors.

\textbf{Practical considerations:} Proposition \ref{prop1} implies that for a given communication environment and under Assumption \ref{assumption1}, there exists a deterministic channel-to-channel mapping function. In other words, given the channels at one antenna set $\mathcal{M}_1$, there is a way to predict exactly the channels at the other antenna set $\mathcal{M}_2$ which could even be a different frequency. In practice, however, there are a few factors that can  add some probabilistic error to this channel prediction such as the measurement noise, the limited ADC bandwidth, and the time-varying channel fading. Evaluating the impact of these practical considerations on the channel prediction error is an important future extension of this work.

\section{Deep Learning Based Channel Mapping: \\ The Main Motivation} \label{sec:DL_Mapping}

In \sref{sec:Exist}, we showed that the channel-to-channel mapping function exists for any given wireless communication environment once the bijectiveness condition in Assumption \ref{assumption1} is satisfied. If this mapping function exists, then how can we model it? The challenge is that this mapping function convolves the various aspects of the communication environment scene, such as the geometry, materials, etc., which makes its analytical characterization  hard--and mostly impossible. In this paper, we propose to leverage deep learning models to \textit{learn} and \textit{approximate} this channel-to-channel mapping function.

For the past three decades, Artificial Neural Networks (ANN) have been known to be good universal function approximators \cite{UnivApprox}. When this capability of neural networks is coupled with the new advances in the field of deep learning  \cite{DL:MethApp}, the result is a powerful tool that could be leveraged to learn challenging complex functions, like that describing the channel-to-channel mapping, $\boldsymbol{\Phi}_{\mathcal{M}_1, f_1 \rightarrow \mathcal{M}_2, f_2}$. More specifically, we propose to approximate $\boldsymbol{\Phi}_{\mathcal{M}_1, f_1 \rightarrow \mathcal{M}_2, f_2}$ using deep neural networks. A deep neural network is a  sequence of parameterized non-linear transformations. Each one roughly represents a single layer in the network, and the stacking of those layers defines the depth of the network. Mathematically, if we adopt an N-layer neural network to learn the channel-to-channel mapping, then the overall function modeled by the neural network can be expressed as a sequence of augmented functions as 
\begin{equation}
\bff_{\text{NN}}(\bh_{u,\mathcal{M}_1(f_1)}, \mathbf W) = \bff_1 \circ \bff_2 \circ \dots \circ \bff_N (\bh_{u,\mathcal{M}_1}(f_1), \mathbf W),
\end{equation}
in which $\mathbf W$ is the set of all parameters of a network, $\bh_{u,\mathcal{M}_1}(f_1)$ is an input vector to that network, and $\bff_{i}(.)$ for $i \in \{1, \dots, N\}$ is a single layer (single non-linear transformation). Our objective will then be to find the neural network weights $\bW$ that make $\bff_{\text{NN}}(.,\bW)$ an accurate approximation of $\boldsymbol{\Phi}_{\mathcal{M}_1, f_1 \rightarrow \mathcal{M}_2, f_2}(.)$, given a certain set of observations (training data points) about the channel vectors at the two antenna sets $\bh_{u,\mathcal{M}_1}(f_1), \bh_{u,\mathcal{M}_2}(f_2)$. In the next section, we will investigate the proposed deep-learning based channel mapping for cell-free massive MIMO systems.    

\section{Deep Learning Based Channel Mapping in Cell-Free Massive MIMO Systems} \label{sec:ML_Sol}
In this section, we investigate the developed deep-learning based channel mapping approach on a cell-free (distributed) massive MIMO setup.  First, we describe the adopted cell-free massive MIMO model in \sref{subsec:cell_free_model} and then we explain in \sref{subsec:sys_operation} how the proposed deep-learning based channel mapping can be applied to this  cell-free massive MIMO setup. Finally, we present the deep learning model in detail in \sref{subsec:DL_model}.

\subsection{Cell-Free Massive MIMO Model} \label{subsec:cell_free_model}

We consider a distributed massive MIMO setup where $M$ geographically-distributed antenna terminals are connected to a central processing unit. Further, we adopt an OFDM system model with $K$ subcarriers, and denote $f_\mathrm{UL}$ and $f_\mathrm{DL}$ as the uplink and downlink center frequencies that could be generally different. In the uplink transmission, if the user is sending a  signal $s_u$, then the received signal by every antenna $m, m=1, 2, ..., M$ at the $k$th subcarrier can be expressed as
\begin{equation}
y_{m,k}=h_{u,m} \left(f_\mathrm{UL}+k \Delta f\right) \ s_u + n_{m,k},
\end{equation}
where $n_m \sim \cN(0,\sigma_{n,m}^2)$ represents the receive noise and $s_u$ satisfies $\bbE[|s_u|^2]=P_u$ with the user transmit power $P_u$.  $\Delta f=\frac{\rm{BW}}{K}$ is the subcarrier spacing frequency  with the uplink bandwidth $\rm{BW}$. For simplicity, we assume that $\mathrm{BW}$ also represents the  downlink bandwidth. In the downlink, If the $M$ antennas are jointly transmitting a signal $s_{M}$, then the received signal at user $u$ and $k$the subcarrier can be written as 
\begin{equation}
y_{u,k}=\sum_{m=1}^{M}  h_{u,m}(f_\mathrm{DL}+k\Delta f)  \ v_{m,k} \ s_{M} + n_{u,k},
\end{equation}
where $n_u \sim \cN(0,\sigma_{n,u}^2)$ represents the receive noise and $s_M$ satisfies $\bbE[|s_M|^2]=\frac{P_\mathrm{T}}{K}$ with the total base station transmit power $P_T$.  The complex wights $v_{m,k} \in \mathbb{C}, \forall m$ denote the beamforming weights at the transmit antennas. For simplicity, we adopt conjugate beamforming where every  beamforming weight $v_{m,k}$ is designed as $v_{m,k}= \frac{1}{\kappa} h^*_{u,m}(f_\mathrm{DL}+k\Delta f)$. The factor $\kappa=\left\| \bh_u(f_\mathrm{DL}+k\Delta f) \right\|_2$ with $\bh_u(f)=[h_{u,1}(f), ...., h_{u,M}(f)]^T$ enforces the power constraint on the  beamforming vector, i.e., it ensures the downlink beamforming vector has a unit norm. With this system model, the downlink achievable rate averaged over the $K$ subcarriers can be expressed as 
\begin{equation}
R=\frac{1}{K} \sum_{k=1}^K \log_2\left(1+\mathsf{SNR} \left\| \bh_u(f_\mathrm{DL}+k\Delta f) \right\|_2^2\right),
\end{equation}
with the signal-to-noise ratio $\mathsf{SNR}=\frac{P_\mathrm{T}}{K \sigma_{n,u}^2}$. Next, we explain the  system operation for the deep-learning based channel mapping, before describing the proposed machine learning model in \sref{subsec:DL_model}.

\subsection{Proposed Channel Mapping with Deep Learning} \label{subsec:sys_operation}
Acquiring the downlink channel knowledge  of the adopted cell-free (distributed) massive MIMO system model normally requires high training and feedback overhead that scales with the number of antennas. In this section, we propose to leverage the interesting channel mapping results in Sections \ref{sec:Exist}-\ref{sec:DL_Mapping} to completely eliminate this overhead. More specifically, let $\mathcal{M}_2$ represent the set of all the $M$ distributed antennas, we propose to learn the mapping from the uplink channels at a subset of the antennas $\mathcal{M}_1 \subset \mathcal{M}_2$ to the downlink channels at all the $M$ antennas, i.e., to the antennas in $\mathcal{M}_2$. \textbf{Learning this mapping means that we only need to train the uplink channels at the antenna subset $\mathcal{M}_1$ and feed these channels to the central processing unit that will use them to directly predict the downlink channels at all the $M$ antennas and at a generally different frequency.} As we show in \sref{sec:Results} for an example of $M=64$ antennas, only $4-8$ antennas in $\mathcal{M}_1$ are enough to have a very good quality for the channel prediction at the $64$ antennas, which significantly reduces the training and feedback overhead as well as the uplink channel feedforward overhead through the fronthaul connecting the distributed antennas and the central unit. The implementation of the proposed deep-learning based channel mapping  entails two modes of operation: learning and prediction. The following is a brief description of both modes.

\textbf{Learning mode:}  In this mode of operation, the adopted distributed massive MIMO system operates as normal, i.e., as it will do if there is no machine learning, and the learning will happen on the side. More specifically, at every coherence time, the uplink and downlink channels will be acquired through two separate and traditional training phases. Then, the uplink channels at the subset of antennas $\mathcal{M}_1$ denoted $\bh_{u,\mathcal{M}_1} \left(f_\mathrm{UL}+k \Delta f\right), \forall k,$ and the downlink channels at all the antennas, denoted $\bh_{u,\mathcal{M}_2} \left(f_\mathrm{DL}+k \Delta f\right), \forall k$, will be saved as a new data point for training the deep learning model. After a sufficient number of data points are collected, the deep learning model (that will be described in detail in \sref{subsec:DL_model}) will be trained to learn the channel-to-channel mapping function from the uplink channels at the antenna subset $\mathcal{M}_1$ to the downlink channels at the $M$ antennas. 

\textbf{Prediction mode:} In this mode of operation, the user needs to send only one uplink pilot that will be used to estimate the uplink channel at the antenna subset $\mathcal{M}_1$. These channels, $\bh_{u,\mathcal{M}_1} \left(f_\mathrm{UL}+k \Delta f\right), \forall k,$, will then be fed forward to the central unit where the trained deep learning model will  be used to directly predict the downlink channels at all the antennas, $\bh_{u,\mathcal{M}_2} \left(f_\mathrm{DL}+k \Delta f\right), \forall k$. These downlink channels will then be used to design the downlink precoding matrix. 

\begin{figure}[t]
	\centering
	\includegraphics[width=0.6\linewidth]{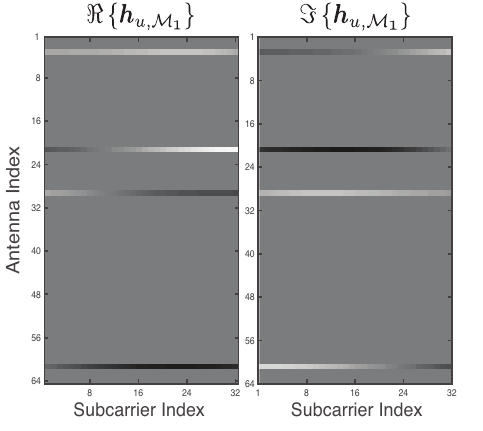}
	\caption{An example of a 3D space-frequency array. The rows with varying brightness represent the channel estimates of the selected antennas across all $K$ sub-carriers while the solid gray rows represent zeroed-out estimates. The x-axis corresponds to subcarrier index, and the y-axis corresponds to antenna index.}
	\label{mask}
\end{figure}

\begin{figure*}[t]
	\centering
	\includegraphics[width=.9\linewidth]{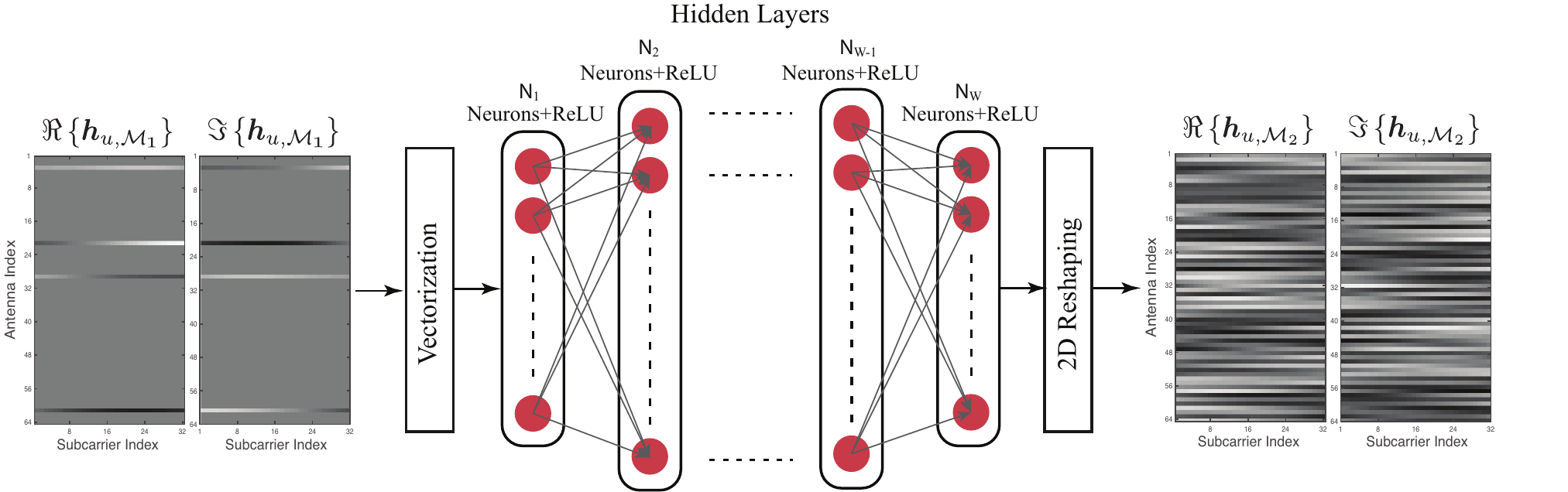}
	\caption{The architecture of the proposed fully-connected network. Every neuron sees all the inputs to its layer, hence the name fully connected. The arrays on the left represent the  channels in set $\mathcal M_1$ while those on the right represent the ``predicted'' channels in set $\mathcal M_2$ which comprises all the $M$ antennas.}
	\label{Arch}
\end{figure*}

\subsection{Deep Learning Model} \label{subsec:DL_model}

Given the proposed system operation above, the following presents a detailed discussion on the deep learning model designed for channel mapping. In particular, it describes the model architecture, the necessary data preprocessing, and how the model is trained.  

\textbf{Neural Network Architecture:} Owing to the spatial and frequency dependencies in the channel estimates, a feedforward fully-connected network, also called Multi-Layer Perceptron (MLP) network \cite{DeepLearning}, represents an interesting choice for the deep learning model. Such neural network type has a very desirable property; each neuron in every layer of the network has access to all outputs from the previous layer--- hence the name. This is expected to help the model learn those channel dependencies and, therefore, become capable of reconstructing the full channel estimates.

The proposed model consists of $W$ layers. Each of them is composed of $N_w$ neurons, where $w\in \{1,\dots,W\}$, followed by $N_w$ non-linearity functions, specifically Rectified Linear Units (ReLU), as depicted in Figure.\ref{Arch}. The model is designed such that it learns three major non-linear transformations. The first squeezes the input into a feature vector with lower dimensionality. This is an attempt to reduce the distortion in the input caused by the missing channel estimates--more on that in the preprocessing subsection next. The second is a sequence of $W-2$ transformations, taking place in what is commonly known as the ``hidden layers''. It starts by taking the first-layer output feature to a high dimensional space, and, then, it progressively passes that new high-dimensional feature through a series of within-space transformations. The goal here is to learn a suitable high-dimensional representation that could help the third, and last, major transformation reconstruct the desired output. That last transformation is basically a projection form the feature space of the hidden layers onto the output space.

\textbf{Preprocessing:}
Getting the deep learning model to perform efficiently, whether in the training or prediction modes, requires some preprocessing of its inputs and target outputs \cite{DeepLearning,LeCun2012}. Data samples usually go through a sequence of preprocessing operations, constituting the preprocessing pipeline, The first two operations, and most commonly used, are centralization and normalization. Channel arrays are centralized using the dataset estimated mean, $\mu$. It is basically the average of all complex values in the dataset. Normalization, on the other hand, is the process of re-scaling the data values to be within the interval $[-1, 1]$. It is performed using the inverse of the maximum absolute value in the dataset. That value is given by:
\begin{equation} \label{MaxAbs}
\Delta = \max_{\substack{ { i = 1, \dotsc, K} \\{j = 1, \dotsc, M }} } \left| h_{i,j} \right|,
\end{equation}
where $h_{i,j}$ is the complex channel estimate of the $i^{th}$ antennal and $k^{th}$ subcarrier, $K$ is the total number of subcarriers of the systems, and finally $M$ is the total number of antennas.

The nature of the problem and the choice of neural network also mandate two more operations in the preprocessing pipeline: array masking (or sampling) and input and output re-shaping. As the core idea of the proposed solution is to perform channel mapping in space and frequency, the channel estimates of all the antennas across all sub-carriers are organized into a 3D channel array where the first two dimensions represent space and frequency and the third represents the real and imaginary components. This array is, then, masked (or smapled) across the space dimension to generate a sparse input array like the one in Figure \ref{mask}. This masking operation retains channel information from a subset of the system antennas while the rest are zeroed-out. This is done using a randomly-generated binary mask with ones across the rows corresponding to the selected antennas. Mathematically, it is expressed as follows:
\begin{equation}
\mathbf h_\mathrm{masked} = \mathbf U\odot \mathbf h_{u,\mathcal{M}_2},
\end{equation}
where $ \mathbf h_{u,\mathcal{M}_2}$ is the input $M \times K \times 2$ array of channels with the third dimension representing the real/imaginary components.  $\mathbf U$ is the $M \times K \times 2$ binary mask, $\odot$ is an element-wise multiplication operator, and $\mathbf h_\mathrm{masked}$ is the $M \times K \times$ masked input.

The second operation in this stage is input and output reshaping. Majority of DL software packages\footnote{These packages are the main development and deployment tools for deep learning models. Some examples include MATLAB deep learning toolbox, TensorFlow and Caffe.} perform real-valued operations, so the complex-valued data samples need to be presented to the network as real values. To meet that requirement, the masked array $\bh_\mathrm{masked}$ is simply flattened to become a single 1-dimensional vector of length $2 K M$. Such choice of reshaping is suitable given the type of the designed neural network. A fully-connected network requires inputs to be in the form of vectors owing to the fact that a single forward pass usually corresponds to a sequence of matrix-multiplication operations \cite{DeepLearning}.

\textbf{Model Training:}
The deep learning model has to go through a training stage, in which it learns the mapping from the input to the output, learning how to reconstruct full channel arrays. The chosen learning approach in this work is supervised learning. It comprises two main components: a training dataset and a loss function. The former is a set of data samples collected from the environment, each of which is a pair of input and its corresponding desired output---hence the name ``supervised learning''. The latter, the loss function, is basically a measure of the model performance during training. The minimization of that loss is an optimization process that amounts to the optimal model parameters. 

It is customary in deep learning problems to use a regression loss function when the model needs to be fitted to some desired response. Since the addressed problem involves channel array reconstruction, the popular Normalized Mean Squared Error (NMSE) is used to assess the output quality. It is given by the following expression:
\begin{equation}\label{instaLoss}
l = \frac{1}{2n}\sum_{i=1}^{n} \frac{||y_{out}-y_{des} ||_{2}^{2}}{||y_{des} ||_2^2},
\end{equation}
where $y_{out},y_{des} \in \mathbb C^n$ are vectorized single model output and its desired response, respectively, and $n$ is the number of antennas multiplied by number of subcarriers. For deep learning models, training is usually performed using small ``batches'' of inputs and their desired responses, commonly called mini-batches, instead of a single pair. This motivates the use of average NMSE:
\begin{equation}
\mathcal{L} = \frac{1}{B} \sum_{i=1}^{B} l_i,
\end{equation}
in which $B$ is the size of the mini-batch.

\section{Experimental Results and Analysis} \label{sec:Results}
In this section, we evaluate the performance of our proposed deep-learning based channel mapping solution for massive MIMO systems and illustrate its potential gains for reducing the channel training and feedback overhead.  We will first describe the adopted scenario and dataset in \sref{subsec:scenario} and the considered deep learning parameters in \sref{subsec:model} before discussing the simulation results in \sref{subsec:sim}. The dataset and code files of this paper are available at \cite{DeepMIMO,MapCode}.

\subsection{Scenario and Dataset} \label{subsec:scenario}
In these experiments, we consider the indoor distributed massive MIMO  scenario 'I1' that is offered by the DeepMIMO dataset \cite{DeepMIMO} and is generated based on the accurate 3D ray-tracing simulator Wireless InSite \cite{Remcom}. This scenario is available at two operating frequencies $2.4$GHz and $2.5$GHz, which emulate the uplink and downlink carrier frequencies of the massive MIMO setup in \sref{subsec:cell_free_model}. As depicted in \figref{env}, the 'I1' scenario comprises a $10$ m $\times 10 $ m room with two tables and $M=64$ antennas tiling up part of the ceiling, i.e. at a height of $2.5$m from the floor. The users are spread across two different x-y grids, each of which is above the floor with a hight of $1$m.

Given this ray-tracing scenario, we generate the DeepMIMO dataset based on the parameters in Table.\ref{EnvPara}. This DeepMIMO dataset constructs the channels between every candidate user location and every antenna terminal at the uplink and downlink frequencies, $2.4$GHz and $2.5$GHz.  
To form training and testing datasets, the elements of the generated DeepMIMO dataset are first shuffled and, then, split into two subsets with 4:1 ratio, namely a training dataset with 80\% of the total size and a testing dataset with the other 20\%. These  datasets  are then used to train the deep learning model and evaluate the performance of the proposed solution.

\begin{table}[t]
\caption{The adopted DeepMIMO dataset parameters}
\centering
\begin{tabular}{|c | c |}
\hline
	Parameter & value \\
	\hline\hline
    Name of scenario & I1-2.4GHz and I1-2.5GHz \\
    \hline
    Number of BSs & 64 \\
    \hline
    Active users & Row 1 to 502 \\
    \hline
    Number of BS antennas in (x, y, x)  & (1,1,1) \\
    \hline
    System BW & 0.02 GHz\\
    \hline
    Number of OFDM sub-carriers & 64\\
    \hline
    OFDM sampling factor & 1 \\
    \hline
    OFDM limit & 16\\
    \hline
\end{tabular}
\label{EnvPara}
\end{table}

\begin{figure}[hbt]
	\centering
	\includegraphics[width=0.8\linewidth]{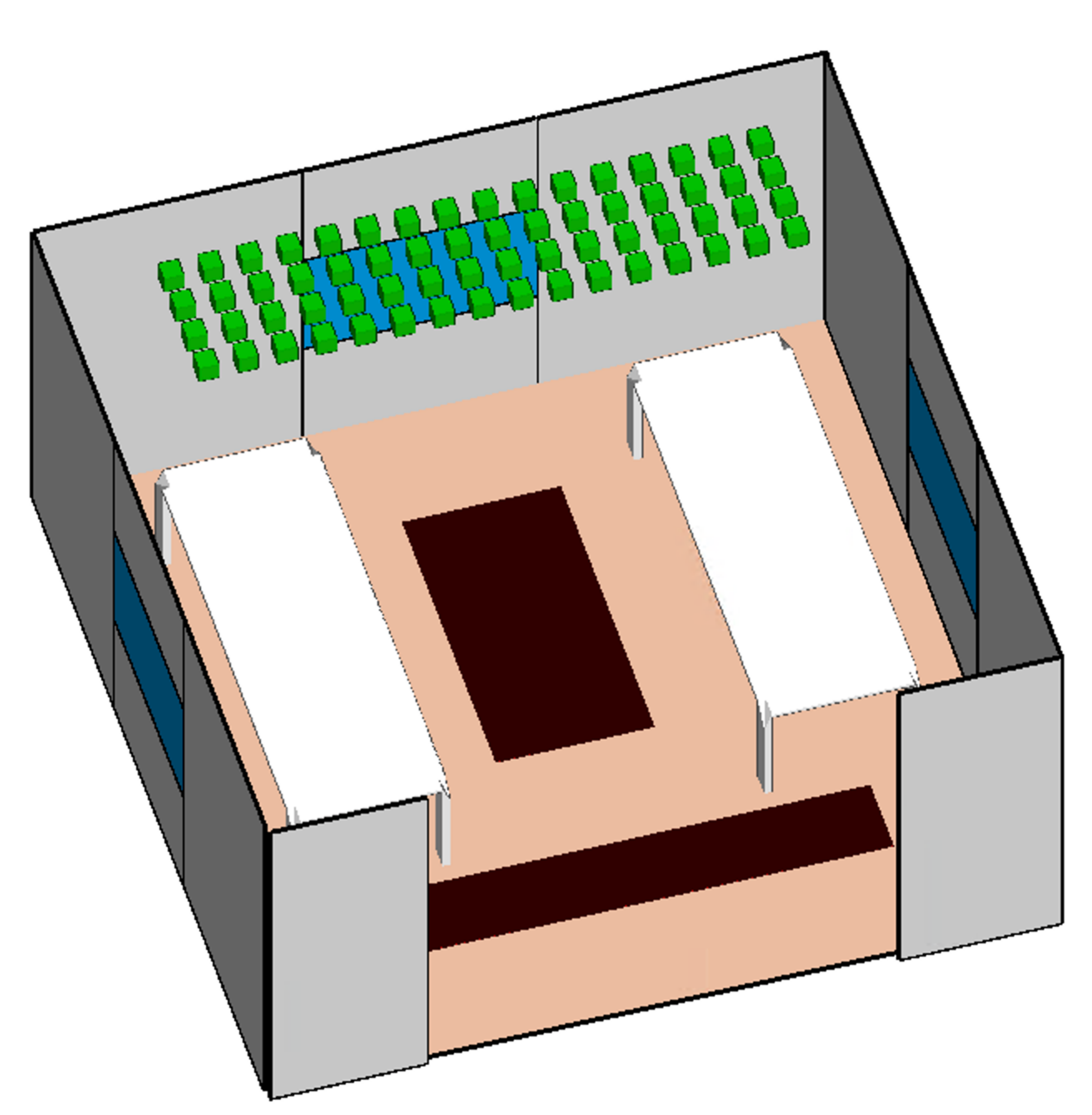}
	\caption{An approximate depiction of the considered environment (scenario). The green little boxes on the ceiling represent the distributed antennas of the base station/access point. The two dark brown rectangles are two grids representing possible user locations. This ray-tracing scenario is constructed using the Wireless InSite by Remcom \cite{Remcom}.}
	\label{env}
\end{figure}

\subsection{Model Training and Testing} \label{subsec:model}
In all experiments in this paper, the fully-connected network used has 4 layers with 1024, 4096, 4096, and 2048 neurons for each layer respectively. This network is trained for about 17 epochs on approximately 121 thousand training samples. Those samples are randomly drawn from the two grids shown in Figure.\ref{env}. The training is performed using ADAptive Moment estimation (ADAM) algorithm \cite{ADAM}---an optimization algorithm with stochastic gradient descent---with learning rate of $1\times10^{-3}$. The network is regularized using L2-norm with weight decay of $1\times10^{-4}$. Following the training, the network performance is tested on approximately 30 thousands \textit{unseen} test samples, also drawn from the same two user grids. The DL modeling, training, and testing  are done using  the MATLAB Deep Learning Toolbox. Code files for the network and its training and testing are available at \cite{MapCode}.

\begin{figure}[t]
	\centering
	\includegraphics[width = 1\linewidth]{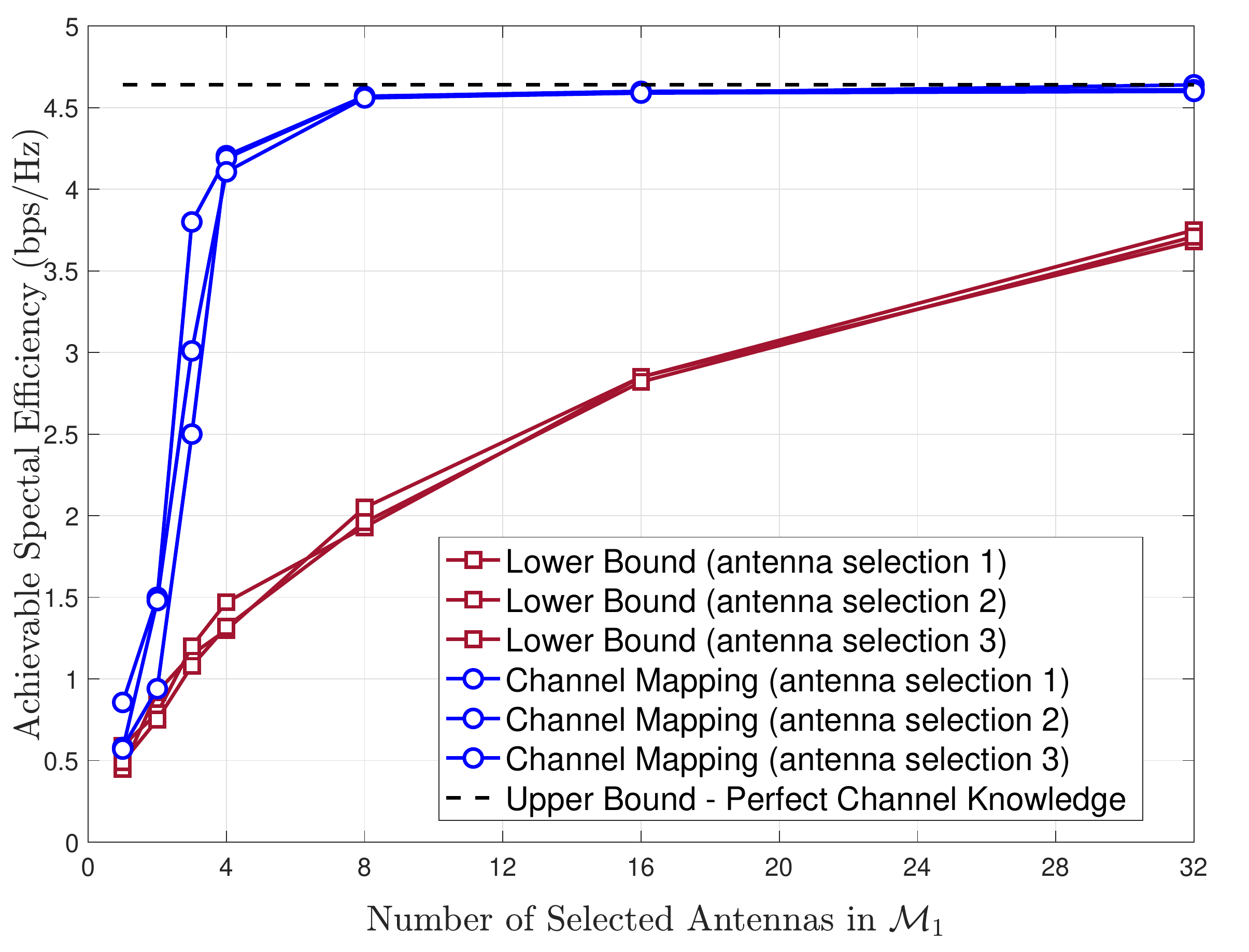}
	\caption{The achievable rates of the proposed deep learning based channel mapping approach versus different numbers of antennas in the subset $\mathcal{M}_1$.  The black dashed line represents the upper bound achieved only with full knowledge of the channels whereas the set of red lines represent the lower bound incurred when only the antennas in the subset $\mathcal{M}_1$ are used.}
	\label{fig1}
\end{figure}

\subsection{Performance Evaluation} \label{subsec:sim}

Considering the distributed massive MIMO setup in \sref{subsec:scenario} and the deep learning model in \sref{subsec:model}, we evaluate the performance of the proposed channel mapping solution as follows. First, for the user positions in the training dataset, we use the DeepMIMO dataset to construct the uplink channels $\bh_{u,\mathcal{M}_1} \left(f_\mathrm{UL}+k \Delta f\right), \forall k$ at the antenna subset $\mathcal{M}_1$, and the  downlink channels $\bh_{u,\mathcal{M}_2} \left(f_\mathrm{DL}+k \Delta f\right), \forall k$ at all the $M=64$ antennas. These channels are then used to train the adopted deep learning model. Note that in this paper, the antennas in the subset $\mathcal{M}_1$ are randomly chosen from the $64$ antennas. Investigating the optimal selection of these antennas conditioning on the given environment is an interesting future research direction.  After training, we use these deep learning model to predict the full downlink channels $\hat{\bh}_{u,\mathcal{M}_2} \left(f_\mathrm{DL}+k \Delta f\right), \forall k$ for the test users given only the uplink channels at the antenna subset $\mathcal{M}_1$. These predicted channels are then used to construct the downlink conjugate beamforming vectors $\bv_k=\hat{\bh}^*_{u,\mathcal{M}_2} \left(f_\mathrm{DL}+k \Delta f\right), \forall k$ with which we calculate the achievable rates using the proposed channel mapping approach. These rates are also compared with (i) the upper bound that represents the rates when the $M$ antennas apply conjugate beamforming with perfect channel knowledge and (ii) the lower bound which is the achievable rate with perfect channel knowledge when only the antennas in the subset $\mathcal{M}_1$ apply conjugate beamforming.

In \figref{fig1}, we consider the indoor massive MIMO setup in \sref{subsec:scenario}, and construct the channels between the users and the antennas using the DeepMIMO dataset \cite{DeepMIMO} considering only the dominant (strongest) path and assuming $f_\mathrm{UL}=f_\mathrm{DL}=2.5$GHz. For this setup, \figref{fig1} compares the achievable rates using the predicted downlink channels with the upper and lower bounds for different numbers of antennas in the subset $\mathcal{M}_1$.  Interestingly, \figref{fig1}  shows that with only 4 antennas, i.e., around $6\%$ of the total number of antennas, the deep learning based predicted (mapped) channels are capable of achieving more than 4 bits/sec/Hz, which brings them within 7\% of the upper bound. This gap falls rapidly and approaches the upper bound with only 8 antennas in the subset $\mathcal{M}_1$. This result highlights the potential gains of the proposed channel mapping approach that significantly reduces the training/feedback overhead. It is worth mentioning at this point that since the antennas in the subset $\mathcal{M}_1$ are randomly selected, we expect the performance of the predicted channels to depend on this selection. As we note from \figref{fig1}, the importance of the specific antenna selection is very high when the number of antennas in $\mathcal{M}_1$ is small. That is very clear form the case of 3 antennas, as the random choice of antennas incurs an achievable rates varying widely between 2.5 to almost 3.7 bits/sec/Hz. Such fluctuation should not be a surprise since the selected antennas could be very close to each other, preventing the network from training with diverse channels.

\begin{figure}[t]
	\includegraphics[width = 1\linewidth]{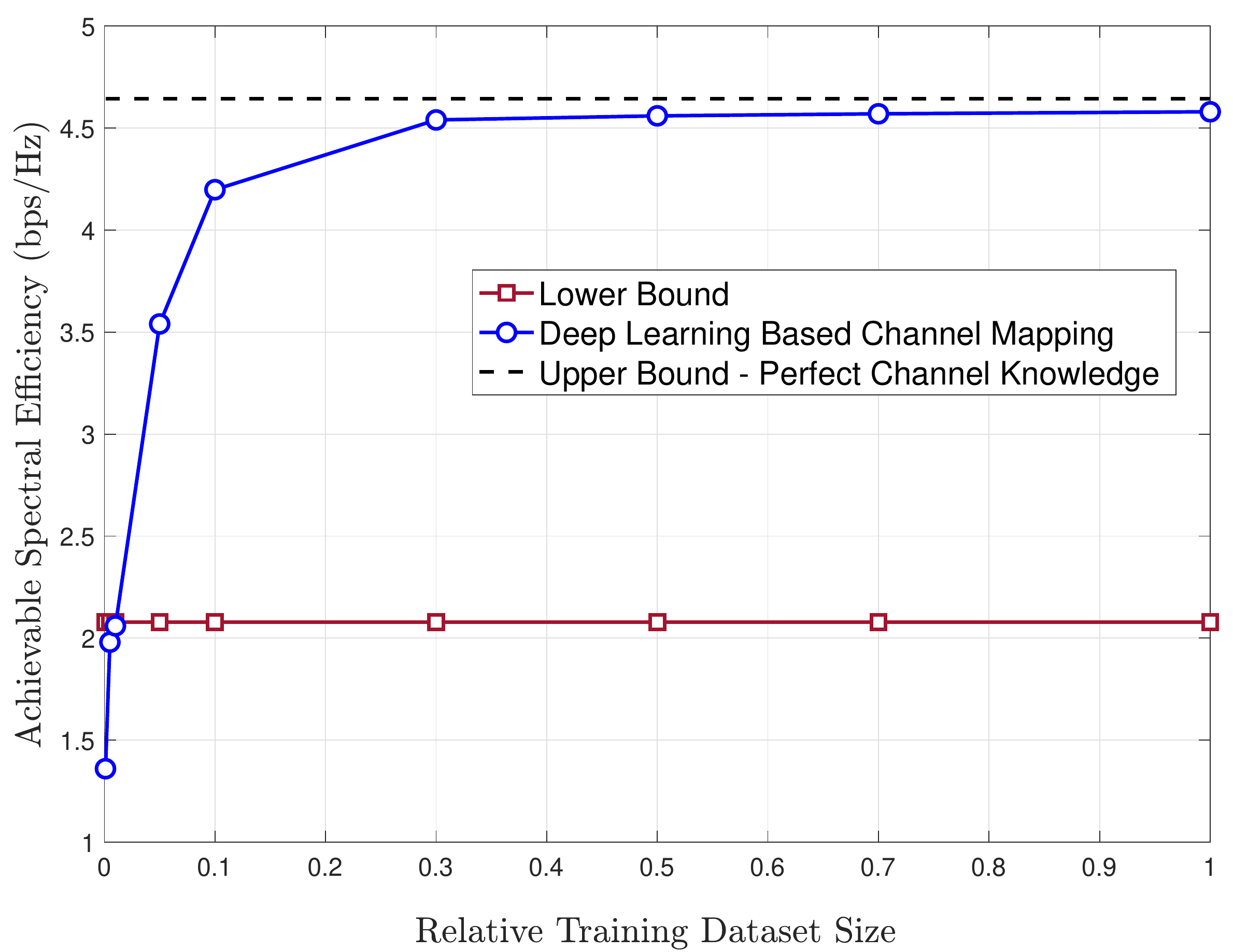}
	\caption{The effect of the training dataset size on the performance of the deep learning based channel mapping approach. The size here is relative to the full dataset size, around 121 thousand data samples.}
	\label{fig2}
\end{figure}

The size of the training dataset is a crucial hyper-parameter for the proposed solution. Figure.\ref{fig2} shows how the performance of the adopted fully-connected neural network improves as the size of the training dataset grows. The case considered here is still the single path case, and the antenna subset $\mathcal{M}_1$ is fixed with $8$ antennas. The network is trained from scratch for every training dataset size and is tested on a fixed-size test set. Both the new training and testing datasets are drawn randomly from the original sets. As expected, the performance improves rapidly as the size increases; however, it is interesting to note that with 8 antennas in $\mathcal{M}_1$, about 30\% of the total training dataset size is enough to approach the upper bound.

Next, we evaluate the performance of the proposed deep-learning based channel mapping approach for the more difficult(and also realistic) case of multi-path channels.  In \figref{fig3}, we consider the massive MIMO setup in \sref{subsec:scenario} with multi-path channels constructed based on the $5$ strongest paths connecting every user and bsse station antenna. \figref{fig3} compares the achievable rates of the proposed channel mapping approach with the upper and lower bounds for two important cases, namely "within-band" and "cross-band". In the "within-band" case, $f_\mathrm{UL}=f_\mathrm{DL}=2.5$GHz, i.e.,  the channel is mapped from the antenna subset $\mathcal{M}_1$ to all the antennas within the same frequency band. \figref{fig3} shows that the performance trend  of the proposed channel mapping approach is similar to that in \figref{fig1}, yet the rate at which the performance improves is not as fast as that in the single-path case. Here, the performance using 4 antennas is within more than 55\% of the upper bound. However, just doubling that number to be 8 antennas boosts the performance to be within 9\% of the upper bound. The upper bound again is approached with merely quarter of the total number of antennas.

\begin{figure}[t]
	\includegraphics[width =\columnwidth]{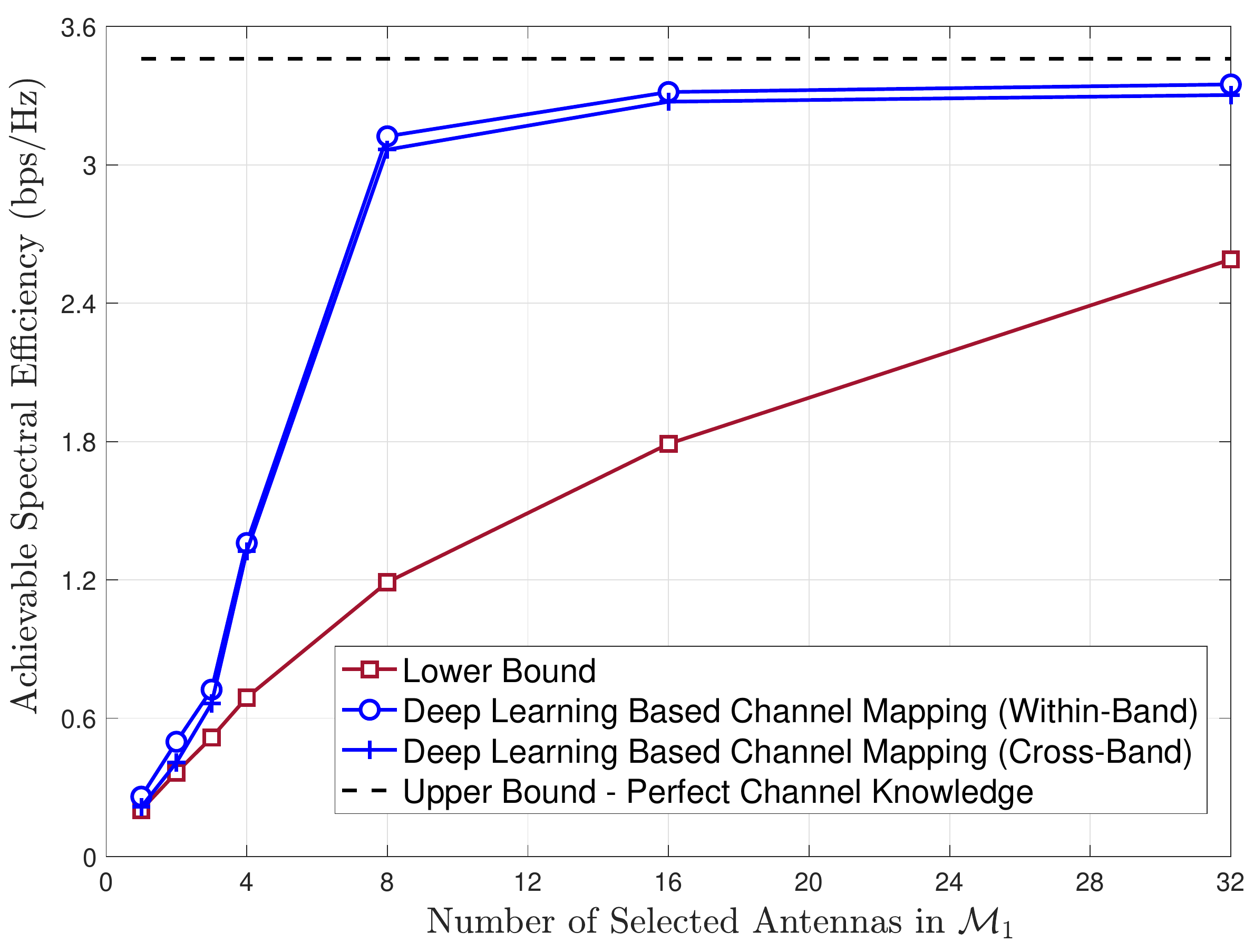}
	\caption{The achievable rates of the proposed deep learning based channel mapping approach versus different numbers of antennas in the subset $\mathcal{M}_1$.  The multi-path channels are constructed using the DeepMIMO dataset \cite{DeepMIMO} considering the strongest $5$paths. The "in-band" scenario represents the case when the channels are mapped from the antennas subset $\mathcal{M}_1$ to all the $64$ antennas at the same frequency band while the "cross-band" case represents the mapping from the uplink channels at  $2.4$GHz and at the antenna subset $\mathcal{M}_1$ to all the $64$ antennas at $2.5$GHz, i.e., the channel mapping is across both space and frequency.}
	\label{fig3}
\end{figure}

In the more challenging "cross-band" case, $f_\mathrm{UL}=2.4$GHz and $f_\mathrm{DL}=2.5$GHz, as the uplink channel at $2.4$GHz and at only a subset of the antennas $\mathcal{M}_1$ are mapped to all the antennas at the downlink frequency $2.5$GHz, i.e. \textit{the channel is mapped across both space and frequency}. This is quite interesting to study since it highlights the ability to predict the downlink channels directly from the uplink channels, which completely eliminates the downlink training/feedback overhead. As shown in \figref{fig3},  the proposed deep-learning based channel mapping approach is clearly capable of learning such mapping very well; as the achievable rate with the predicted downlink channels  converges very closely to the upper bound with only $8-16$ antennas. 

\section{Conclusion} \label{sec:Conclusion}

In this paper, we introduced the new concept of channel mapping in space and frequency, where the channels are mapped from one set of antennas at one frequency band to another set of antennas at another frequency band. For any given environment, we first proved that this channel-to-channel mapping function exists under the condition that the mapping from the candidate user locations to the channels in the first set of antennas is bijective---a condition that can be satisfied with high probability in several practical communication scenarios. Then, we proposed to leverage deep neural networks to learn this complex channel-to-channel mapping function that involves the various elements of the environment. For FDD/TDD cell-free massive MIMO systems, we applied the channel mapping concept to reduce the downlink training/feedback overhead as well as the fronthaul signaling overhead. Extensive simulation results based on accurate 3D ray-tracing highlighted the potential gains of the channel mapping approach. For example, for a cell-free massive MIMO setup of 64 distributed antennas, the results show that the channels at only 4-8 antennas are required to be able to efficiently predict the channels at all the 64 antennas and achieve near-optimal data rates compared to perfect channel knowledge.  This highlights the potential gains of the proposed channel mapping approach in significantly reducing the training/feedback overhead in massive MIMO systems. For future work, it is interesting to investigate the impact of various practical considerations, such as the measurement noise, the limited ADC bandwidth, and the time-varying channel fading, on the performance of the proposed channel mapping approach.

\balance


\begin{thebibliography}{10}
	\providecommand{\url}[1]{#1}
	\csname url@samestyle\endcsname
	\providecommand{\newblock}{\relax}
	\providecommand{\bibinfo}[2]{#2}
	\providecommand{\BIBentrySTDinterwordspacing}{\spaceskip=0pt\relax}
	\providecommand{\BIBentryALTinterwordstretchfactor}{4}
	\providecommand{\BIBentryALTinterwordspacing}{\spaceskip=\fontdimen2\font plus
		\BIBentryALTinterwordstretchfactor\fontdimen3\font minus
		\fontdimen4\font\relax}
	\providecommand{\BIBforeignlanguage}[2]{{%
			\expandafter\ifx\csname l@#1\endcsname\relax
			\typeout{** WARNING: IEEEtran.bst: No hyphenation pattern has been}%
			\typeout{** loaded for the language `#1'. Using the pattern for}%
			\typeout{** the default language instead.}%
			\else
			\language=\csname l@#1\endcsname
			\fi
			#2}}
	\providecommand{\BIBdecl}{\relax}
	\BIBdecl
	
	\bibitem{HeathJr2016}
	R.~W. Heath, N.~Gonzlez-Prelcic, S.~Rangan, W.~Roh, and A.~M. Sayeed, ``An
	overview of signal processing techniques for millimeter wave {MIMO}
	systems,'' \emph{IEEE Journal of Selected Topics in Signal Processing},
	vol.~10, no.~3, pp. 436--453, April 2016.
	
	\bibitem{Boccardi2014}
	F.~Boccardi, R.~Heath, A.~Lozano, T.~Marzetta, and P.~Popovski, ``Five
	disruptive technology directions for {5G},'' \emph{IEEE Communications
		Magazine}, vol.~52, no.~2, pp. 74--80, Feb. 2014.
	
	\bibitem{11ad}
	\BIBentryALTinterwordspacing
	{IEEE 802.11ad}, ``{IEEE} 802.11ad standard draft {D0.1}.'' [Online].
	Available: \url{www.ieee802.org/11/Reports/tgad update.htm}
	\BIBentrySTDinterwordspacing
	
	\bibitem{Andrews2014}
	J.~Andrews, S.~Buzzi, W.~Choi, S.~Hanly, A.~Lozano, A.~Soong, and J.~Zhang,
	``What will {5G} be?'' \emph{IEEE Journal on Selected Areas in
		Communications}, vol.~32, no.~6, pp. 1065--1082, June 2014.
	
	\bibitem{Sanguinetti2019}
	L.~{Sanguinetti}, E.~{Bjornson}, and J.~{Hoydis}, ``{Towards Massive MIMO 2.0:
		Understanding spatial correlation, interference suppression, and pilot
		contamination},'' \emph{arXiv e-prints}, p. arXiv:1904.03406, Apr 2019.
	
	\bibitem{Bjoernson2016}
	E.~Bjornson, E.~G. Larsson, and T.~L. Marzetta, ``Massive {MIMO}: Ten myths and
	one critical question,'' \emph{IEEE Communications Magazine}, vol.~54, no.~2,
	pp. 114--123, February 2016.
	
	\bibitem{Vasisht2016}
	\BIBentryALTinterwordspacing
	D.~Vasisht, S.~Kumar, H.~Rahul, and D.~Katabi, ``Eliminating channel feedback
	in next-generation cellular networks,'' in \emph{Proceedings of the 2016 ACM
		SIGCOMM Conference}, ser. SIGCOMM '16.\hskip 1em plus 0.5em minus 0.4em\relax
	New York, NY, USA: ACM, 2016, pp. 398--411. [Online]. Available:
	\url{http://doi.acm.org/10.1145/2934872.2934895}
	\BIBentrySTDinterwordspacing
	
	\bibitem{Rottenberg2019}
	F.~{Rottenberg}, R.~{Wang}, J.~{Zhang}, and A.~F. {Molisch}, ``{Channel
		Extrapolation in FDD Massive MIMO: Theoretical Analysis and Numerical
		Validation},'' \emph{arXiv e-prints}, p. arXiv:1902.06844, Feb 2019.
	
	\bibitem{Ali2018}
	A.~{Ali}, N.~{Gonzalez-Prelcic}, and R.~W. {Heath}, ``Millimeter wave
	beam-selection using out-of-band spatial information,'' \emph{IEEE
		Transactions on Wireless Communications}, vol.~17, no.~2, pp. 1038--1052, Feb
	2018.
	
	\bibitem{Maschietti2019}
	F.~{Maschietti}, D.~{Gesbert}, and P.~{de Kerret}, ``{Coordinated Beam
		Selection in Millimeter Wave Multi-User MIMO Using Out-of-Band
		Information},'' \emph{arXiv e-prints}, p. arXiv:1903.12589, Mar 2019.
	
	\bibitem{Alkhateeb2018}
	A.~Alkhateeb, S.~Alex, P.~Varkey, Y.~Li, Q.~Qu, and D.~Tujkovic, ``Deep
	learning coordinated beamforming for highly-mobile millimeter wave systems,''
	\emph{IEEE Access}, vol.~6, pp. 37\,328--37\,348, 2018.
	
	\bibitem{Li2018a}
	X.~Li, A.~Alkhateeb, and C.~Tepedelenlio{\u{g}}lu, ``Generative adversarial
	estimation of channel covariance in vehicular millimeter wave systems,''
	2018.
	
	\bibitem{Alkhateeb18a}
	I.~B. A.~Alkhateeb and S.~Alex, ``Machine learning for reliable mmwave systems:
	Blockage prediction and proactive handoff,'' in \emph{)}, Nov 2018.
	
	\bibitem{Taha2019}
	A.~{Taha}, M.~{Alrabeiah}, and A.~{Alkhateeb}, ``{Enabling Large Intelligent
		Surfaces with Compressive Sensing and Deep Learning},'' \emph{arXiv
		e-prints}, p. arXiv:1904.10136, Apr 2019.
	
	\bibitem{Vieira2017}
	J.~{Vieira}, E.~{Leitinger}, M.~{Sarajlic}, X.~{Li}, and F.~{Tufvesson}, ``Deep
	convolutional neural networks for massive mimo fingerprint-based
	positioning,'' in \emph{2017 IEEE 28th Annual International Symposium on
		Personal, Indoor, and Mobile Radio Communications (PIMRC)}, Oct 2017, pp.
	1--6.
	
	\bibitem{Savic2015}
	V.~{Savic} and E.~G. {Larsson}, ``Fingerprinting-based positioning in
	distributed massive {MIMO} systems,'' in \emph{2015 IEEE 82nd Vehicular
		Technology Conference (VTC2015-Fall)}, Sep. 2015, pp. 1--5.
	
	\bibitem{UnivApprox}
	K.~Hornik, M.~Stinchcombe, and H.~White, ``Multilayer feedforward networks are
	universal approximators,'' \emph{Neural networks}, vol.~2, no.~5, pp.
	359--366, 1989.
	
	\bibitem{DL:MethApp}
	L.~Deng, D.~Yu \emph{et~al.}, ``Deep learning: methods and applications,''
	\emph{Foundations and Trends{\textregistered} in Signal Processing}, vol.~7,
	no. 3--4, pp. 197--387, 2014.
	
	\bibitem{DeepLearning}
	I.~Goodfellow, Y.~Bengio, and A.~Courville, \emph{Deep learning}.\hskip 1em
	plus 0.5em minus 0.4em\relax MIT press, 2016.
	
	\bibitem{LeCun2012}
	Y.~A. LeCun, L.~Bottou, G.~B. Orr, and K.-R. M{\"u}ller, ``Efficient
	backprop,'' in \emph{Neural networks: Tricks of the trade}.\hskip 1em plus
	0.5em minus 0.4em\relax Springer, 2012, pp. 9--48.
	
	\bibitem{DeepMIMO}
	A.~Alkhateeb, ``{DeepMIMO}: A generic deep learning dataset for millimeter wave
	and massive {MIMO} applications,'' in \emph{Proc. of Information Theory and
		Applications Workshop (ITA)}, San Diego, CA, Feb 2019, pp. 1--8.
	
	\bibitem{MapCode}
	\BIBentryALTinterwordspacing
	[Online]. Available:
	\url{https://github.com/malrabeiah/DL-Massive-MIMO/tree/master/ChannelMapping}
	\BIBentrySTDinterwordspacing
	
	\bibitem{Remcom}
	Remcom, ``Wireless insite,'' \url{http://www.remcom.com/wireless-insite}.
	
	\bibitem{ADAM}
	D.~P. Kingma and J.~Ba, ``Adam: A method for stochastic optimization,''
	\emph{arXiv preprint arXiv:1412.6980}, 2014.
	
\end{thebibliography}
\end{document}